\documentclass[pra,twocolumn, preprintnumbers,amsmath,amssymb,nobalancelastpage,longbibliography]{revtex4-1}

\usepackage{graphicx, color, graphpap}
\usepackage{enumitem}
\usepackage{amssymb}
\usepackage{amsthm}
\usepackage{float}
\usepackage{multirow}
\usepackage[colorlinks=true,citecolor=blue,linkcolor=magenta]{hyperref}
\usepackage[T1]{fontenc}
\usepackage{bbm}
\usepackage{thmtools,thm-restate}
\usepackage{verbatim}
\usepackage{mathtools}

\long\def\ca#1\cb{} 

\newcommand{\ket}[1]{|#1\rangle}               
\newcommand{\bra}[1]{\langle #1|}              
\newcommand{\dya}[1]{\ket{#1}\!\bra{#1}}
\newcommand{\ip}[2]{\langle #1|#2\rangle}      

\newcommand{\rank}{\text{rank}}

\newcommand{\LL}{\text{L}}
\newcommand{\HS}{\text{HS}}

\newcommand{\Tr}{{\rm Tr}}

\renewcommand{\geq}{\geqslant}
\renewcommand{\leq}{\leqslant}

\newcommand{\ad}{^\dagger}

\newcommand*{\id}{\openone}

{}
{}
\newtheorem{theorem}{Theorem}
\newtheorem{lemma}{Lemma}
\newtheorem{proposition}{Proposition}
\newtheorem{example}{Example}

\begin{document}
\title{Strong bound between trace distance and Hilbert-Schmidt distance for low-rank states}
\author{Patrick J. Coles}
\affiliation{Theoretical Division, Los Alamos National Laboratory, Los Alamos, NM 87545, USA}
\author{M. Cerezo}
\affiliation{Theoretical Division, Los Alamos National Laboratory, Los Alamos, NM 87545, USA}
\author{Lukasz Cincio}
\affiliation{Theoretical Division, Los Alamos National Laboratory, Los Alamos, NM 87545, USA}

\begin{abstract}
The trace distance between two quantum states, $\rho$ and $\sigma$, is an operationally meaningful quantity in quantum information theory. However, in general it is difficult to compute, involving the diagonalization of $\rho - \sigma$. In contrast, the Hilbert-Schmidt distance can be computed without diagonalization, although it is less operationally significant. Here, we relate the trace distance and the Hilbert-Schmidt distance with a bound that is particularly strong when either $\rho$ or $\sigma$ is low rank. Our bound is stronger than the bound one could obtain via the norm equivalence of the Frobenius and trace norms. We also consider bounds that are useful not only for low-rank states but also for low-entropy states. Our results have relevance to quantum information theory, quantum algorithms design, and quantum complexity theory.
\end{abstract}

\maketitle

\section{Introduction}

The trace distance is employed in the definition of security in quantum cryptography \cite{horodecki2005,renner2005}, is related to the error probability in quantum state discrimination \cite{helstrom1969}, and generally has operational relevance to many quantum information protocols \cite{nielsen2010,wilde2017}. It is defined by
\begin{equation}
D(\rho,\sigma)=\frac{1}{2}\Tr|\rho-\sigma| = \frac{1}{2} \|\rho - \sigma\|_1\,,
\end{equation}
where $\|M\|_1 = \Tr\sqrt{M\ad M}$ is called the 1-norm or trace norm. Computing the trace distance may, in general, involve diagonalizing the matrix $\Delta = \rho - \sigma$, which can be exponentially difficult as a result of the exponentially large dimension of the density matrix representation.

An alternative metric is the Hilbert-Schmidt distance:
\begin{equation}
D_{\HS}(\rho,\sigma)=\Tr[(\rho-\sigma)^2] =  \|\rho - \sigma\|_2^2\,,
\end{equation}
where $\|M\|_2 = \sqrt{\Tr (M\ad M)}$ is called the 2-norm or Frobenius norm. While the Hilbert-Schmidt distance does not share the operational relevance of the trace distance \cite{ozawa2000entanglement}, it has the benefit that one can compute it without doing matrix diagonalization. 
Moreover, if $\rho$ and $\sigma$ are in a quantum form (i.e. were prepared on a quantum computer) it is known that their Hilbert-Schmidt distance is efficiently computable (logarithmic in matrix dimension) on a quantum computer \cite{garcia2013swap,cincio2018learning}. Because of the latter, the Hilbert-Schmidt distance is employed as a cost function in recent variational hybrid quantum-classical algorithms \cite{larose2018,arrasmith2018, Khatri_LaRose_Poremba_Cincio_Sornborger_Coles_2018,cerezo2019variational}.

In many applications, one is interested in upper-bounding the trace distance between two states, i.e., showing that $D(\rho,\sigma)\leq \epsilon$ for some small number $\epsilon$. Hence, in this work, we explore upper bounds on $D(\rho,\sigma)$, and in particular, upper bounds formulated in terms of $D_{\HS}(\rho,\sigma)$, since the latter may be computable.

For example, the norm equivalence \cite{van1983matrix} of the trace norm and Frobenius norm can be written as
\begin{equation}
\label{eqnNormEQ0}
\|M\|_2 \leq \|M\|_1 \leq \sqrt{\rank(M)}\|M\|_2\,,
\end{equation}
where the first inequality follows from the monotonicity of the Schatten norms \cite{raissouli2010various,horn1990matrix} and the second one follows from the Cauchy-Schwartz inequality. Applying this to $M =\Delta$ gives
\begin{equation}
\label{eqnNormEQ1}
(1/4)D_{\HS}(\rho,\sigma)\leq D(\rho,\sigma)^2 \leq (d/4)\cdot D_{\HS}(\rho,\sigma)\,,
\end{equation}
where $d$ is the Hilbert space dimension, and we used the inequality $\rank(\Delta)\leq d$, which is the best one can do without any prior knowledge about $\Delta$. As $d$ grows exponentially in the number of quantum subsystems, $d$ can be very large, making the upper bound in \eqref{eqnNormEQ1} very weak. Hence it is natural to ask whether the upper bound in \eqref{eqnNormEQ1} can be tightened for certain states $\rho$ and $\sigma$.

The goal of this article is to explore possible tightenings of this upper bound when $\rho$ and/or $\sigma$ are low-rank states or low-entropy states. Such tightening would be expected since the extreme case of pure states, $\rho = \dya{\psi}$ and $\sigma = \dya{\phi}$, gives $D(\rho,\sigma)^2 = 1-|\ip{\psi}{\phi}|^2 $ and $D_{\HS}(\rho,\sigma) = 2-2|\ip{\psi}{\phi}|^2 $. Hence for pure states we have
\begin{equation}
\label{eqnPureStates}
D(\rho,\sigma)^2 = (1/2)D_{\HS}(\rho,\sigma)\,, \end{equation}
and no dimension-dependent factor appears here.

Our main result allows us to write the bound
\begin{equation}
\label{eqnMainResultIntro}
    (1/2)D_{\HS}(\rho,\sigma)\leq D(\rho,\sigma)^2 \leq  R\cdot D_{\HS}(\rho,\sigma)\,,
\end{equation}
where $R= \rank(\rho)\rank(\sigma)/[\rank(\rho)+\rank(\sigma)]$. The slight tightening of the lower bound 
relative to \eqref{eqnNormEQ1} is straightforward to show~\cite{coles2012unification}, and hence our contribution here is the upper bound in \eqref{eqnMainResultIntro}. This implies that, even when one state is full rank, if the other state is low rank, then the Hilbert-Schmidt distance is of the same order-of-magnitude as the square of the trace distance. Under these conditions, one can use the (easy-to-calculate) Hilbert-Schmidt distance as a surrogate for the (operationally meaningful) trace distance.

The applications of this result include:
\begin{enumerate}
    \item In quantum information protocols, such as quantum cryptographic protocols \cite{horodecki2005,renner2005}, the trace distance is often employed as an operational figure-of-merit. The Hilbert-Schmidt distance is easier to calculate both analytically and numerically and hence can be used to upper bound the trace distance for these protocols, provided one has a low-rank guarantee.
    \item The Hilbert-Schmidt distance is often employed as a cost function in variational hybrid quantum-classical algorithms (often with the low-rank assumption already made) \cite{larose2018,arrasmith2018, Khatri_LaRose_Poremba_Cincio_Sornborger_Coles_2018,cerezo2019variational}, since it can be efficiently computed on a quantum computer. Because of our bound, minimizing the cost in these algorithms implies that one is also minimizing a function of the trace distance. Hence our bound gives further justification to these algorithms.
    \item The trace distance plays an important role in quantum complexity theory, since the problem of deciding whether the trace distance is large or small is QSZK-complete 
    (QSZK = quantum statistical zero knowledge) \cite{watrous2002quantum}. QSZK is a complexity class that contains BQP (BQP = bounded-error quantum polynomial time),
    $$\text{BQP}\subseteq \text{QSZK}\,.$$ 
    It is generally believed that BQP is not equal to QSZK, and hence deciding if two states are close or far in trace distance 
    cannot (in general) be efficiently done on a quantum computer \cite{watrous2002quantum}. If there exist states (e.g., low-rank states) for which the trace distance is approximately equal to a quantity that can be efficiently computed on a quantum computer (i.e., Hilbert-Schmidt distance), then this would have implications in quantum complexity theory. (See Conclusions for elaboration.)
\end{enumerate}

In what follows, we first discuss tightenings of \eqref{eqnNormEQ1} for low-rank states, including our main result of an upper bound that is stronger than what one could obtain via norm equivalence. Then we explore some bounds for the more general scenario of low-entropy states.

\section{Rank-based bounds}

Consider the following useful lemma.
\begin{lemma}
\label{lemma1}
Let $\Delta = \rho - \sigma = \Delta_+ - \Delta_-$. Here $\Delta_+$ and $\Delta_-$ respectively correspond to the positive and negative part of $\Delta$, with $\Delta_+\geq0$, $\Delta_-\geq0$, and $\Delta_+\Delta_-=0$. Then we have:
\begin{align}
\rank(\Delta_+) \leq \rank(\rho)\,,\\
\rank(\Delta_-) \leq \rank(\sigma)\,.
\end{align}
\end{lemma}
\begin{proof}
Let $\{r_j\}$, $\{s_j\}$, and $\{\delta_j\}$ respectively denote the eigenvalues of $\rho$, $\sigma$, and $\Delta$, where the eigenvalues in each set are ordered in decreasing order. Weyl's inequality \cite{horn1990matrix,weyl1912asymptotische} applied to $\rho = \Delta + \sigma$ gives
\begin{equation}
    r_j \geq \delta_j + s_d \,,\quad\forall j\,.
\end{equation}
Because $\sigma\geq 0$, we have $s_d\geq 0$, and hence
\begin{equation}
    r_j \geq \delta_j\,,\quad\forall j\,.
\end{equation}
Since the eigenvalues of $\rho$ are bigger than the eigenvalues of $\Delta$, then when $r_j =0$ we must have $\delta_j \leq 0$. This implies that $\rank(\rho) \geq \rank(\Delta_+)$\,.

Similarly, we can define $\overline{\Delta} =  \sigma - \rho =  \Delta_- - \Delta_+$, where $\Delta_-$ corresponds to the positive part of $\overline{\Delta}$. Writing $\sigma = \overline{\Delta} + \rho$, and applying Weyl's inequality gives 
\begin{equation}
    s_j \geq \overline{\delta}_j\,,\quad\forall j\,,
\end{equation}
where $\{\overline{\delta}_j\}$ are the ordered eigenvalues of $\overline{\Delta}$. Since the eigenvalues of $\sigma$ are bigger than those of $\overline{\Delta}$, then when $s_j =0$ we must have $\overline{\delta}_j \leq 0$, implying $\rank(\sigma) \geq \rank(\Delta_-)$\,.
\end{proof}

Combining this lemma with the norm equivalence in \eqref{eqnNormEQ0} gives the following result.

\begin{proposition}
For any two quantum states $\rho$ and $\sigma$,
\begin{equation}
\label{eqnBound1}
D(\rho,\sigma)^2 \leq \left(\frac{\rank(\rho) + \rank(\sigma)}{4}\right)D_{\HS}(\rho,\sigma)\,.
\end{equation}
\end{proposition}
\begin{proof}
From Lemma~\ref{lemma1}, we have
\begin{align}
\rank(\Delta) &= \rank(\Delta_+)+\rank(\Delta_-)\\
&\leq \rank(\rho)+\rank(\sigma)\,.
\end{align}
We remark that this inequality can also be obtained from the fact that the rank is subadditive, i.e., $\rank(A+B) \leq \rank(A) + \rank(B)$. The upper bound in \eqref{eqnNormEQ0} then gives:
\begin{align}
4D(\rho,\sigma)^2 &\leq \rank(\Delta)D_{\HS}(\rho,\sigma)\\
&\leq (\rank(\rho)+\rank(\sigma))D_{\HS}(\rho,\sigma)\,.
\end{align}
\end{proof}

Equation~\eqref{eqnBound1} gets rid of the dimension-dependent factor in \eqref{eqnNormEQ1} and replaces it with a rank-based factor. This can potentially improve the bound, and indeed one can see from \eqref{eqnPureStates} that \eqref{eqnBound1} is tight when both $\rho$ and $\sigma$ are pure. On the other hand, there are many other cases where \eqref{eqnBound1} is loose. For example, as we will soon see from our main result below, the following inequality holds when, say, $\rho$ is pure but $\sigma$ is an arbitrary state (i.e., when only one of the two states is pure):
\begin{equation}
\label{eqnPureRho}
D(\rho,\sigma)^2 \leq D_{\HS}(\rho,\sigma)\,.
\end{equation}
In this special case, $\rank(\sigma)$ can grow in proportion to the dimension $d$ and hence the bound in \eqref{eqnBound1} can be extremely loose.

The looseness of \eqref{eqnBound1} motivates looking for a tightening of the bound that is not based on the norm equivalence in \eqref{eqnNormEQ0}. Indeed, that is what we do in the following theorem, which is our main result.

\begin{theorem}
\label{thm1}
For any two quantum states $\rho$ and $\sigma$,
\begin{equation}
\label{eqnBound2}
D(\rho,\sigma)^2 \leq R\cdot D_{\HS}(\rho,\sigma)\,,
\end{equation}
where we refer to $R$ as the reduced rank (defined analogously to the reduced mass in physics):
\begin{equation}
\label{eqnReducedRank}
R = \frac{\rank(\rho)\rank(\sigma)}{\rank(\rho) + \rank(\sigma)}\,.
\end{equation}
\end{theorem}
\begin{proof}
Let $\tau_+ = \Delta_+ / \Tr(\Delta_+)$ and $\tau_- = \Delta_- / \Tr(\Delta_-)$\,. Note that both $\tau_+$ and $\tau_-$ are density matrices (positive semidefinite with trace one). Since the purity of a density matrix is lower bounded by the inverse of its rank, we have
\begin{align}
\Tr(\tau_+^2) \geq 1/\rank(\Delta_+) \geq 1/\rank(\rho)\,,\\
\Tr(\tau_-^2) \geq 1/\rank(\Delta_-) \geq 1/\rank(\sigma)\,,
\end{align}
where we have employed Lemma~\ref{lemma1}. Using the definitions of $\tau_+$ and $\tau_-$, we have
\begin{align}
\Tr(\Delta_+^2) \geq \Tr(\Delta_+)^2 /\rank(\rho)\,,\\
\Tr(\Delta_-^2) \geq \Tr(\Delta_-)^2/\rank(\sigma)\,.
\end{align}
Summing these two inequalities gives:
\begin{align}
\Tr(\Delta_+^2)+\Tr(\Delta_-^2) \geq \frac{\Tr(\Delta_+)^2}{\rank(\rho)}+\frac{\Tr(\Delta_-)^2}{\rank(\sigma)}\,.
\end{align}
The left-hand-side of this inequality is $D_{\HS}(\rho,\sigma)$, while $\Tr(\Delta_+) = \Tr(\Delta_-) = D(\rho,\sigma)$. Hence we have
\begin{align}
D_{\HS}(\rho,\sigma) \geq D(\rho,\sigma)^2\left(\frac{1}{\rank(\rho)}+\frac{1}{\rank(\sigma)}\right)\,,
\end{align}
which is equivalent to \eqref{eqnBound2}.
\end{proof}

Let us now show that \eqref{eqnBound2} is stronger than \eqref{eqnBound1}. First note that
\begin{equation}
2\rank(\rho)\rank(\sigma)\leq \rank(\rho)^2+\rank(\sigma)^2\,,    
\end{equation}
which implies that
\begin{equation}
4\rank(\rho)\rank(\sigma)\leq (\rank(\rho)+\rank(\sigma))^2\,.
\end{equation}
Dividing through by $4(\rank(\rho)+\rank(\sigma))$ gives
\begin{equation}
R \leq (\rank(\rho)+\rank(\sigma))/4\,,    
\end{equation}
which shows that \eqref{eqnBound2} implies \eqref{eqnBound1}.

We also remark that the reduced rank has the property
\begin{equation}
    R \leq \min \{\rank(\rho),\rank(\sigma)\}\,,
\end{equation}
and hence provides a stronger bound than simply taking the minimum of the two ranks.

Next let us consider some examples. For these examples, it will be helpful to define 
\begin {equation}
Q(\rho ,\sigma)  = D(\rho,\sigma)^2 / D_{\HS}(\rho,\sigma)\,,
\end{equation}
which is well-defined so long as $\rho \neq \sigma$.

\begin{example}
Consider the example noted above in \eqref{eqnPureRho}, where $\rho$ is pure. In this case, the reduced rank is
\begin{equation}
R = \frac{\rank(\sigma)}{1+\rank(\sigma)} \leq 1\,,    
\end{equation}
and hence \eqref{eqnBound2} implies \eqref{eqnPureRho}. This is a dramatic improvement compared to the bound in \eqref{eqnBound1}. 
\end{example}

\begin{example}
Suppose that $\rho = \Pi /r$ is proportional to a rank $r$ projector $\Pi$ and $\sigma = \id /d$ is maximally mixed. In this case, 
$$D(\rho,\sigma)= \frac{d-r}{d},\hspace{8pt} R=\frac{dr}{d+r},\hspace{8pt}D_{\HS}(\rho,\sigma)= \frac{d-r}{dr}\,.$$
This gives
\begin{equation}
Q(\rho, \sigma)=R(d^2-r^2) / d^2\,,    
\end{equation}
where the right-hand-side is approximately $R$ when $r \ll d$, and hence \eqref{eqnBound2} becomes tight in this limit.
\end{example}

\begin{example}
Suppose that $\rho = \Pi_r /r$ and $\sigma = \Pi_s / s$, where $\Pi_r$ and $\Pi_s$ are orthogonal projectors ($\Pi_r \Pi_s = 0$) and their ranks are $r$ and $s$, respectively. In this case, 
$$D(\rho,\sigma)= 1,\hspace{8pt} D_{\HS}(\rho,\sigma)= \frac{r+s}{rs},\hspace{8pt}Q(\rho, \sigma ) =R\,.$$
Hence, \eqref{eqnBound2} is perfectly tight for these states. Note that the orthogonality of $\rho$ and $\sigma$ in this example means that this example only applies to $R\leq d/4$.
\end{example}

Example~3 demonstrates that \eqref{eqnBound2} can be satisfied with equality for all values of $R$ in the range $1/2 \leq R \leq d/4$. For $R > d/4$, the upper bound in \eqref{eqnNormEQ1} is obviously tighter than \eqref{eqnBound2}. Therefore, in general, one can use the tighter of these two bounds:
\begin{equation}
\label{eqnQbound}
    1/2 \leq Q(\rho,\sigma) \leq \min\{ R, d/4\}\,.
\end{equation}

\begin{figure}[t]
    \centering
    \includegraphics[width=.75\columnwidth]{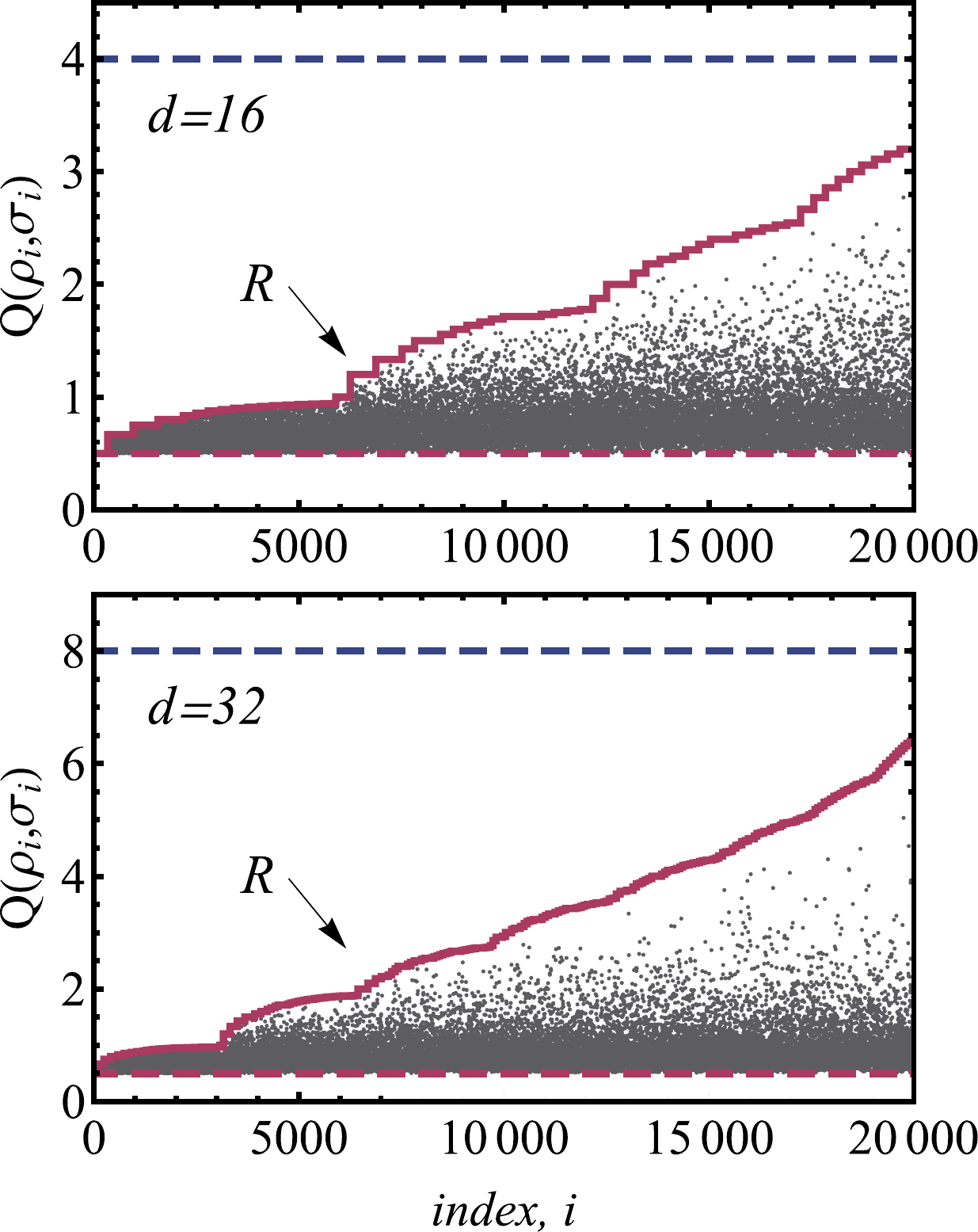}
    \caption{Bounds for $Q(\rho_i,\sigma_i)=D(\rho_i,\sigma_i)^2/D_{\HS}(\rho_i,\sigma_i)$ for $d=16$ (top) and $d=32$ (bottom). The gray dots correspond to $Q(\rho_i,\sigma_i)$ for random states $\sigma_i$ with uniformly distributed rank and purity and random states $\rho_i$ with bounded rank ($1\leq \rank(\rho)\leq d/4$). Results were ordered by increasing value of $R$, with the value of $R$ shown as the solid red curve. The dashed blue line is the alternative upper bound $Q(\rho,\sigma)\leq d/4$, obtained from the norm equivalence, while the dashed red line is the lower bound  $1/2\leq Q(\rho,\sigma)$ from \eqref{eqnMainResultIntro}. Over this range of $R$ values, our bound $Q(\rho,\sigma)\leq R$ can be saturated and is tighter than the bound from norm equivalence.}
    \label{fig:1}
\end{figure}

Figure \ref{fig:1} shows representative numerical results comparing \eqref{eqnBound2} with the upper bound in \eqref{eqnNormEQ1}. In particular we plot $Q(\rho,\sigma)$
for $20000$ random states $\sigma$ (with uniformly distributed rank and purity) and random states $\rho$ with $1\leq \rank(\rho)\leq d/4$, for $d=16,32$. Results are ordered by increasing value of $R$. For all points, our bound is tighter than the one obtained from norm equivalence. Furthermore, one can see that our bound can be saturated over this range of $R$ values.

\section{Entropy-based bounds}

The rank of a matrix is an abruptly changing function of its eigenvalues. So it makes sense to consider a smoother function such as entropy. While there are many such entropy functions, we consider the linear entropy
\begin{equation}
    S_{\LL}(\rho) = 1-\Tr(\rho^2)
\end{equation}
for two reasons: it is the natural entropy to associate with the Hilbert-Schmidt distance, and it is efficient to compute. This is in the spirit that we would like to upper-bound the trace distance with a quantity that is easily computable.

Consider the following bound.

\begin{proposition}
\label{prop2}
For any two quantum states $\rho$ and $\sigma$,
\begin{equation}
\label{eqnBound3}
    D(\rho,\sigma)^2\leq  \frac{1}{2}[D_{\HS}(\rho,\sigma)+S_{\LL}(\rho)+S_{\LL}(\sigma)]\,.
\end{equation}
\end{proposition}
\begin{proof}
Let $\{\delta^{+}_j\}$ and $\{\delta^{-}_j\}$ respectively be the eigenvalues of $\Delta_+$ and $\Delta_-$. Since $D(\rho,\sigma) =\Tr(\Delta_+) = \Tr(\Delta_-)$, we can write
\begin{align}
    D(\rho,\sigma)^2 &= \frac{1}{2}[\Tr(\Delta_+)^2+\Tr(\Delta_-)^2]\\
    &=\frac{1}{2}\left[\sum_{jk}\delta^{+}_j\delta^{+}_k+\sum_{jk}\delta^{-}_j\delta^{-}_k\right]\\
    &=\frac{1}{2}\left[D_{\HS}(\rho,\sigma)+\sum_{k\neq j}\delta^{+}_j\delta^{+}_k+\sum_{k\neq j}\delta^{-}_j\delta^{-}_k\right]\,,
\end{align}
where we used the fact that $D_{\HS}(\rho,\sigma) = \sum_j [(\delta^{+}_j)^2+(\delta^{-}_j)^2]$. Next we use Weyl's inequality (see proof of Lemma~\ref{lemma1}), which implies that:
\begin{align}
\label{eqnProof122}
r_j &\geq \delta^{+}_j,\quad\forall j\\
s_j &\geq \delta^{-}_j,\quad\forall j\,,
\end{align}
where we assume the eigenvalues are ordered in decreasing order. Since these eigenvalues are non-negative, we have
\begin{align}
\label{eqnProof123}
    D(\rho,\sigma)^2 &\leq \frac{1}{2}\left[D_{\HS}(\rho,\sigma)+\sum_{k\neq j}r_j r_k+\sum_{k\neq j}s_j s_k\right]\,.
\end{align}
Finally, note that
\begin{align}
\label{eqnProof124}
S_{\LL}(\rho) = (\Tr(\rho))^2 - \Tr(\rho^2) = \sum_{k\neq j}r_j r_k,
\end{align}
and similarly $S_{\LL}(\sigma) = \sum_{k\neq j}s_j s_k $. Plugging these relations into \eqref{eqnProof123} proves the result.
\end{proof}

Equation~\eqref{eqnBound3} is tight when both $\rho$ and $\sigma$ are pure, since $S_{\LL}(\rho)=S_{\LL}(\sigma)=0$ in this case. On the other hand, when only $\rho$ is pure, the bound in \eqref{eqnBound3} becomes $(1/2)[D_{HS}(\rho,\sigma)+S_{\LL}(\sigma)]$, which does not reduce to \eqref{eqnPureRho}. In this case, \eqref{eqnBound3} can be either stronger or weaker than \eqref{eqnPureRho}.

The following is an alternative bound.

\begin{proposition}
\label{prop3}
For any two quantum states $\rho$ and $\sigma$,
\begin{equation}
\label{eqnBound4}
    D(\rho,\sigma)^2\leq  D_{\HS}(\rho,\sigma)+\min \{ S_{\LL}(\rho),S_{\LL}(\sigma)\}\,.
\end{equation}
\end{proposition}
\begin{proof}
We will prove that
\begin{equation}
\label{eqnProof55}
    D(\rho,\sigma)^2\leq  D_{\HS}(\rho,\sigma)+ S_{\LL}(\rho)\,.
\end{equation}
Then by symmetry the roles of $\rho$ and $\sigma$ can be interchanged to give:
\begin{equation}
\label{eqnProof56}
    D(\rho,\sigma)^2\leq  D_{\HS}(\rho,\sigma)+ S_{\LL}(\sigma)\,.
\end{equation}
Taking the minimum of the two bounds in \eqref{eqnProof55} and \eqref{eqnProof56} then gives \eqref{eqnBound4}.

To prove \eqref{eqnProof55}, we follow a similar approach as the proof of Proposition~\ref{prop2}. Namely, we write
\begin{align}
    D(\rho,\sigma)^2 &= \Tr(\Delta_+)^2\\
    &=\sum_{jk}\delta^{+}_j\delta^{+}_k\\
    &=\sum_{j}(\delta^{+}_j)^2 +\sum_{k\neq j}\delta^{+}_j\delta^{+}_k\\
    &\leq D_{\HS}(\rho,\sigma)+\sum_{k\neq j}\delta^{+}_j\delta^{+}_k\,.
\end{align}
Then from \eqref{eqnProof122} we have 
\begin{align}
    D(\rho,\sigma)^2 
    &\leq D_{\HS}(\rho,\sigma)+\sum_{k\neq j}r_j r_k\,,
\end{align}
and inserting \eqref{eqnProof124} we obtain \eqref{eqnProof55}.
\end{proof}

Equation~\eqref{eqnBound4} does not have the factor of (1/2) in front of $D_{\HS}(\rho,\sigma)$ like \eqref{eqnBound3} does. On the other hand, it does reduce to \eqref{eqnPureRho} when $\rho$ is pure. 

There are cases in which \eqref{eqnBound3} is stronger than \eqref{eqnBound4}, and there are cases where the reverse is true. Hence, in general, one can take the minimum of the two bounds from \eqref{eqnBound3} and \eqref{eqnBound4} for our strongest entropy-based bound.

We do remark on the following issue. Both \eqref{eqnBound3} and \eqref{eqnBound4} are additive bounds, where the entropy is added to the Hilbert-Schmidt distance. This is in contrast to our rank-based bounds, which were multiplicative. Typically multiplicative bounds are more desirable than additive bounds, since one would like the bound to go to zero when the Hilbert-Schmidt distance goes to zero. For this reason, we propose that future work is needed to search for multiplicative entropy-based bounds, where the entropy term multiplies the Hilbert-Schmidt distance. The existence of such bounds remains an interesting open question \footnote{For example, the following bound would be natural conjecture:  $D(\rho,\sigma)^2 \leq D_{\HS}(\rho,\sigma)/ \Tr(\rho^2)$. However, we numerically found this bound to be violated.}.

\section{Conclusions}

Low-rank and low-entropy quantum states show up naturally in, for example, condensed matter physics, where bipartite cuts of condensed matter ground states lead to weakly entangled subsystems \cite{li2008}. Low-rank states also appear in data science, where the covariance matrix has only a small number of principal components due to redundant features \cite{bishop2006pattern, murphy2012machine}. These applications motivate the work in this article, where we gave improved upper bounds on the trace distance for low-rank or low-entropy states. We focused on upper bounds involving the Hilbert-Schmidt distance, since this quantity can be efficiently computed on a quantum computer. Hence, our bounds may benefit the field of quantum algorithms, by making it easier for quantum algorithms to provide a tight bound on the operationally relevant trace distance.

Our main result, Theorem~\ref{thm1}, gave a bound involving the reduced rank (a quantity analogous to the reduced mass in physics). This bound was stronger than what one could obtain directly from the equivalence of the Frobenius and trace norms. We also gave bounds involving the linear entropy in Propositions~\ref{prop2} and \ref{prop3}, where the linear entropy appeared as an additive term in the bound. An open question is whether multiplicative entropy-based bounds exist, and we believe this an important direction for future work.

Finally, as noted in the Introduction, deciding whether the trace distance is large or small 
is QSZK-complete, and hence is believed to be difficult for a quantum computer. Nevertheless, this does not preclude the possiblity that the trace distance for a special class of states might be efficiently estimated with a quantum computer. Because of our results, we speculate that this might be true for low-rank states, although we leave a formal proof for future work. If true, this would suggest that estimating the trace distance for some states lies outside of BQP, while for other states it lies inside of BQP. 

\section{Acknowledgements}

We thank Rolando Somma and Marco Tomamichel for helpful conversations. All authors acknowledge support from the LDRD program at Los Alamos National Laboratory (LANL). PJC additionally acknowledges support from the LANL ASC Beyond Moore's Law project. MC was also supported by the Center for Nonlinear Studies at LANL. LC was also supported by the DOE through the J. Robert Oppenheimer fellowship.


\begin{thebibliography}{22}%
\makeatletter
\providecommand \@ifxundefined [1]{%
 \@ifx{#1\undefined}
}%
\providecommand \@ifnum [1]{%
 \ifnum #1\expandafter \@firstoftwo
 \else \expandafter \@secondoftwo
 \fi
}%
\providecommand \@ifx [1]{%
 \ifx #1\expandafter \@firstoftwo
 \else \expandafter \@secondoftwo
 \fi
}%
\providecommand \natexlab [1]{#1}%
\providecommand \enquote  [1]{``#1''}%
\providecommand \bibnamefont  [1]{#1}%
\providecommand \bibfnamefont [1]{#1}%
\providecommand \citenamefont [1]{#1}%
\providecommand \href@noop [0]{\@secondoftwo}%
\providecommand \href [0]{\begingroup \@sanitize@url \@href}%
\providecommand \@href[1]{\@@startlink{#1}\@@href}%
\providecommand \@@href[1]{\endgroup#1\@@endlink}%
\providecommand \@sanitize@url [0]{\catcode `\\12\catcode `\$12\catcode
  `\&12\catcode `\#12\catcode `\^12\catcode `\_12\catcode `\%12\relax}%
\providecommand \@@startlink[1]{}%
\providecommand \@@endlink[0]{}%
\providecommand \url  [0]{\begingroup\@sanitize@url \@url }%
\providecommand \@url [1]{\endgroup\@href {#1}{\urlprefix }}%
\providecommand \urlprefix  [0]{URL }%
\providecommand \Eprint [0]{\href }%
\providecommand \doibase [0]{http://dx.doi.org/}%
\providecommand \selectlanguage [0]{\@gobble}%
\providecommand \bibinfo  [0]{\@secondoftwo}%
\providecommand \bibfield  [0]{\@secondoftwo}%
\providecommand \translation [1]{[#1]}%
\providecommand \BibitemOpen [0]{}%
\providecommand \bibitemStop [0]{}%
\providecommand \bibitemNoStop [0]{.\EOS\space}%
\providecommand \EOS [0]{\spacefactor3000\relax}%
\providecommand \BibitemShut  [1]{\csname bibitem#1\endcsname}%
\let\auto@bib@innerbib\@empty
\bibitem [{\citenamefont {Ben-Or}\ \emph {et~al.}(2005)\citenamefont {Ben-Or},
  \citenamefont {Horodecki}, \citenamefont {Leung}, \citenamefont {Mayers},\
  and\ \citenamefont {Oppenheim}}]{horodecki2005}%
  \BibitemOpen
  \bibfield  {author} {\bibinfo {author} {\bibfnamefont {M.}~\bibnamefont
  {Ben-Or}}, \bibinfo {author} {\bibfnamefont {M.}~\bibnamefont {Horodecki}},
  \bibinfo {author} {\bibfnamefont {D~W.}\ \bibnamefont {Leung}}, \bibinfo
  {author} {\bibfnamefont {D.}~\bibnamefont {Mayers}}, \ and\ \bibinfo {author}
  {\bibfnamefont {J.}~\bibnamefont {Oppenheim}},\ }\bibfield  {title} {\enquote
  {\bibinfo {title} {The universal composable security of quantum key
  distribution},}\ }in\ \href {https://doi.org/10.1007/978-3-540-30576-7_21}
  {\emph {\bibinfo {booktitle} {Theory of Cryptography}}},\ \bibinfo {editor}
  {edited by\ \bibinfo {editor} {\bibfnamefont {J.}~\bibnamefont {Kilian}}}\
  (\bibinfo  {publisher} {Springer Berlin Heidelberg},\ \bibinfo {address}
  {Berlin, Heidelberg},\ \bibinfo {year} {2005})\ pp.\ \bibinfo {pages}
  {386--406}\BibitemShut {NoStop}%
\bibitem [{\citenamefont {Renner}\ and\ \citenamefont
  {K{\"o}nig}(2005)}]{renner2005}%
  \BibitemOpen
  \bibfield  {author} {\bibinfo {author} {\bibfnamefont {R.}~\bibnamefont
  {Renner}}\ and\ \bibinfo {author} {\bibfnamefont {R.}~\bibnamefont
  {K{\"o}nig}},\ }\bibfield  {title} {\enquote {\bibinfo {title} {Universally
  composable privacy amplification against quantum adversaries},}\ }in\ \href
  {https://doi.org/10.1007/978-3-540-30576-7_22} {\emph {\bibinfo {booktitle}
  {Theory of Cryptography}}},\ \bibinfo {editor} {edited by\ \bibinfo {editor}
  {\bibfnamefont {J.}~\bibnamefont {Kilian}}}\ (\bibinfo  {publisher} {Springer
  Berlin Heidelberg},\ \bibinfo {address} {Berlin, Heidelberg},\ \bibinfo
  {year} {2005})\ pp.\ \bibinfo {pages} {407--425}\BibitemShut {NoStop}%
\bibitem [{\citenamefont {Helstrom}(1969)}]{helstrom1969}%
  \BibitemOpen
  \bibfield  {author} {\bibinfo {author} {\bibfnamefont {Carl~W.}\ \bibnamefont
  {Helstrom}},\ }\bibfield  {title} {\enquote {\bibinfo {title} {Quantum
  detection and estimation theory},}\ }\href {\doibase 10.1007/BF01007479}
  {\bibfield  {journal} {\bibinfo  {journal} {J. Stat. Phys.}\ }\textbf
  {\bibinfo {volume} {1}},\ \bibinfo {pages} {231--252} (\bibinfo {year}
  {1969})}\BibitemShut {NoStop}%
\bibitem [{\citenamefont {{Nielsen}}\ and\ \citenamefont
  {{Chuang}}(2010)}]{nielsen2010}%
  \BibitemOpen
  \bibfield  {author} {\bibinfo {author} {\bibfnamefont {M.~A.}\ \bibnamefont
  {{Nielsen}}}\ and\ \bibinfo {author} {\bibfnamefont {I.~L.}\ \bibnamefont
  {{Chuang}}},\ }\href@noop {} {\emph {\bibinfo {title} {Quantum Computation
  and Quantum Information, Cambridge University Press}}}\ (\bibinfo {year}
  {2010})\BibitemShut {NoStop}%
\bibitem [{\citenamefont {Wilde}(2017)}]{wilde2017}%
  \BibitemOpen
  \bibfield  {author} {\bibinfo {author} {\bibfnamefont {M.~M.}\ \bibnamefont
  {Wilde}},\ }\href {\doibase 10.1017/9781316809976} {\emph {\bibinfo {title}
  {Quantum Information Theory}}},\ \bibinfo {edition} {2nd}\ ed.\ (\bibinfo
  {publisher} {Cambridge University Press},\ \bibinfo {year}
  {2017})\BibitemShut {NoStop}%
\bibitem [{\citenamefont {Ozawa}(2000)}]{ozawa2000entanglement}%
  \BibitemOpen
  \bibfield  {author} {\bibinfo {author} {\bibfnamefont {M.}~\bibnamefont
  {Ozawa}},\ }\bibfield  {title} {\enquote {\bibinfo {title} {Entanglement
  measures and the {H}ilbert--{S}chmidt distance},}\ }\href {\doibase
  10.1016/S0375-9601(00)00171-7} {\bibfield  {journal} {\bibinfo  {journal}
  {Phys. Lett. A}\ }\textbf {\bibinfo {volume} {268}},\ \bibinfo {pages}
  {158--160} (\bibinfo {year} {2000})}\BibitemShut {NoStop}%
\bibitem [{\citenamefont {Garcia-Escartin}\ and\ \citenamefont
  {Chamorro-Posada}(2013)}]{garcia2013swap}%
  \BibitemOpen
  \bibfield  {author} {\bibinfo {author} {\bibfnamefont {J.~C.}\ \bibnamefont
  {Garcia-Escartin}}\ and\ \bibinfo {author} {\bibfnamefont {P.}~\bibnamefont
  {Chamorro-Posada}},\ }\bibfield  {title} {\enquote {\bibinfo {title} {Swap
  test and {Hong-Ou-Mandel} effect are equivalent},}\ }\href {\doibase
  10.1103/PhysRevA.87.052330} {\bibfield  {journal} {\bibinfo  {journal} {Phys.
  Rev. A}\ }\textbf {\bibinfo {volume} {87}},\ \bibinfo {pages} {052330}
  (\bibinfo {year} {2013})}\BibitemShut {NoStop}%
\bibitem [{\citenamefont {{Cincio}}\ \emph {et~al.}(2018)\citenamefont
  {{Cincio}}, \citenamefont {{Suba{\c{s}}{\i}}}, \citenamefont {{Sornborger}},\
  and\ \citenamefont {{Coles}}}]{cincio2018learning}%
  \BibitemOpen
  \bibfield  {author} {\bibinfo {author} {\bibfnamefont {L.}~\bibnamefont
  {{Cincio}}}, \bibinfo {author} {\bibfnamefont {Y.}~\bibnamefont
  {{Suba{\c{s}}{\i}}}}, \bibinfo {author} {\bibfnamefont {A.~T.}\ \bibnamefont
  {{Sornborger}}}, \ and\ \bibinfo {author} {\bibfnamefont {P.~J.}\
  \bibnamefont {{Coles}}},\ }\bibfield  {title} {\enquote {\bibinfo {title}
  {{Learning the quantum algorithm for state overlap}},}\ }\href {\doibase
  10.1088/1367-2630/aae94a} {\bibfield  {journal} {\bibinfo  {journal} {New J.
  Phys.}\ }\textbf {\bibinfo {volume} {20}},\ \bibinfo {eid} {113022} (\bibinfo
  {year} {2018})}\BibitemShut {NoStop}%
\bibitem [{\citenamefont {LaRose}\ \emph {et~al.}(2019)\citenamefont {LaRose},
  \citenamefont {Tikku}, \citenamefont {O'Neel-Judy}, \citenamefont {Cincio},\
  and\ \citenamefont {Coles}}]{larose2018}%
  \BibitemOpen
  \bibfield  {author} {\bibinfo {author} {\bibfnamefont {Ryan}\ \bibnamefont
  {LaRose}}, \bibinfo {author} {\bibfnamefont {Arkin}\ \bibnamefont {Tikku}},
  \bibinfo {author} {\bibfnamefont {{\'E}tude}\ \bibnamefont {O'Neel-Judy}},
  \bibinfo {author} {\bibfnamefont {Lukasz}\ \bibnamefont {Cincio}}, \ and\
  \bibinfo {author} {\bibfnamefont {Patrick~J}\ \bibnamefont {Coles}},\
  }\bibfield  {title} {\enquote {\bibinfo {title} {Variational quantum state
  diagonalization},}\ }\href
  {https://www.nature.com/articles/s41534-019-0167-6} {\bibfield  {journal}
  {\bibinfo  {journal} {npj Quantum Information}\ }\textbf {\bibinfo {volume}
  {5}},\ \bibinfo {pages} {8} (\bibinfo {year} {2019})}\BibitemShut {NoStop}%
\bibitem [{\citenamefont {Arrasmith}\ \emph {et~al.}(2019)\citenamefont
  {Arrasmith}, \citenamefont {Cincio}, \citenamefont {Sornborger},
  \citenamefont {Zurek},\ and\ \citenamefont {Coles}}]{arrasmith2018}%
  \BibitemOpen
  \bibfield  {author} {\bibinfo {author} {\bibfnamefont {Andrew}\ \bibnamefont
  {Arrasmith}}, \bibinfo {author} {\bibfnamefont {Lukasz}\ \bibnamefont
  {Cincio}}, \bibinfo {author} {\bibfnamefont {Andrew~T}\ \bibnamefont
  {Sornborger}}, \bibinfo {author} {\bibfnamefont {Wojciech~H}\ \bibnamefont
  {Zurek}}, \ and\ \bibinfo {author} {\bibfnamefont {Patrick~J}\ \bibnamefont
  {Coles}},\ }\bibfield  {title} {\enquote {\bibinfo {title} {Variational
  consistent histories as a hybrid algorithm for quantum foundations},}\ }\href
  {https://www.nature.com/articles/s41467-019-11417-0} {\bibfield  {journal}
  {\bibinfo  {journal} {Nature Communications}\ }\textbf {\bibinfo {volume}
  {10}},\ \bibinfo {pages} {3438} (\bibinfo {year} {2019})}\BibitemShut
  {NoStop}%
\bibitem [{\citenamefont {Khatri}\ \emph {et~al.}(2019)\citenamefont {Khatri},
  \citenamefont {LaRose}, \citenamefont {Poremba}, \citenamefont {Cincio},
  \citenamefont {Sornborger},\ and\ \citenamefont
  {Coles}}]{Khatri_LaRose_Poremba_Cincio_Sornborger_Coles_2018}%
  \BibitemOpen
  \bibfield  {author} {\bibinfo {author} {\bibfnamefont {Sumeet}\ \bibnamefont
  {Khatri}}, \bibinfo {author} {\bibfnamefont {Ryan}\ \bibnamefont {LaRose}},
  \bibinfo {author} {\bibfnamefont {Alexander}\ \bibnamefont {Poremba}},
  \bibinfo {author} {\bibfnamefont {Lukasz}\ \bibnamefont {Cincio}}, \bibinfo
  {author} {\bibfnamefont {Andrew~T}\ \bibnamefont {Sornborger}}, \ and\
  \bibinfo {author} {\bibfnamefont {Patrick~J}\ \bibnamefont {Coles}},\
  }\bibfield  {title} {\enquote {\bibinfo {title} {Quantum-assisted quantum
  compiling},}\ }\href {https://quantum-journal.org/papers/q-2019-05-13-140/}
  {\bibfield  {journal} {\bibinfo  {journal} {Quantum}\ }\textbf {\bibinfo
  {volume} {3}},\ \bibinfo {pages} {140} (\bibinfo {year} {2019})}\BibitemShut
  {NoStop}%
\bibitem [{\citenamefont {Cerezo}\ \emph {et~al.}(2019)\citenamefont {Cerezo},
  \citenamefont {Poremba}, \citenamefont {Cincio},\ and\ \citenamefont
  {Coles}}]{cerezo2019variational}%
  \BibitemOpen
  \bibfield  {author} {\bibinfo {author} {\bibfnamefont {M}~\bibnamefont
  {Cerezo}}, \bibinfo {author} {\bibfnamefont {Alexander}\ \bibnamefont
  {Poremba}}, \bibinfo {author} {\bibfnamefont {Lukasz}\ \bibnamefont
  {Cincio}}, \ and\ \bibinfo {author} {\bibfnamefont {Patrick~J}\ \bibnamefont
  {Coles}},\ }\bibfield  {title} {\enquote {\bibinfo {title} {Variational
  quantum fidelity estimation},}\ }\href {https://arxiv.org/abs/1906.09253}
  {\bibfield  {journal} {\bibinfo  {journal} {arXiv preprint arXiv:1906.09253}\
  } (\bibinfo {year} {2019})}\BibitemShut {NoStop}%
\bibitem [{\citenamefont {Van~Loan}\ and\ \citenamefont
  {Golub}(1983)}]{van1983matrix}%
  \BibitemOpen
  \bibfield  {author} {\bibinfo {author} {\bibfnamefont {C.~F.}\ \bibnamefont
  {Van~Loan}}\ and\ \bibinfo {author} {\bibfnamefont {G.~H.}\ \bibnamefont
  {Golub}},\ }\href@noop {} {\emph {\bibinfo {title} {Matrix Computations}}}\
  (\bibinfo  {publisher} {Johns Hopkins University Press},\ \bibinfo {year}
  {1983})\BibitemShut {NoStop}%
\bibitem [{\citenamefont {Ra{\i}ssouli}\ and\ \citenamefont
  {Jebril}(2010)}]{raissouli2010various}%
  \BibitemOpen
  \bibfield  {author} {\bibinfo {author} {\bibfnamefont {M.}~\bibnamefont
  {Ra{\i}ssouli}}\ and\ \bibinfo {author} {\bibfnamefont {I.~H.}\ \bibnamefont
  {Jebril}},\ }\bibfield  {title} {\enquote {\bibinfo {title} {Various proofs
  for the decrease monotonicity of the {S}chatten's power norm, various
  families of $\mathbb{R}^n$-norms and some open problems},}\ }\href@noop {}
  {\bibfield  {journal} {\bibinfo  {journal} {Int. J. Open Problems Compt.
  Math}\ }\textbf {\bibinfo {volume} {3}} (\bibinfo {year} {2010})}\BibitemShut
  {NoStop}%
\bibitem [{\citenamefont {Horn}\ and\ \citenamefont
  {Johnson}(1990)}]{horn1990matrix}%
  \BibitemOpen
  \bibfield  {author} {\bibinfo {author} {\bibfnamefont {R.~A.}\ \bibnamefont
  {Horn}}\ and\ \bibinfo {author} {\bibfnamefont {C.~R.}\ \bibnamefont
  {Johnson}},\ }\href@noop {} {\emph {\bibinfo {title} {Matrix Analysis}}}\
  (\bibinfo  {publisher} {Cambridge University Press},\ \bibinfo {year}
  {1990})\BibitemShut {NoStop}%
\bibitem [{\citenamefont {Coles}(2012)}]{coles2012unification}%
  \BibitemOpen
  \bibfield  {author} {\bibinfo {author} {\bibfnamefont {Patrick~J}\
  \bibnamefont {Coles}},\ }\bibfield  {title} {\enquote {\bibinfo {title}
  {Unification of different views of decoherence and discord},}\ }\href
  {https://journals.aps.org/pra/abstract/10.1103/PhysRevA.85.042103} {\bibfield
   {journal} {\bibinfo  {journal} {Physical Review A}\ }\textbf {\bibinfo
  {volume} {85}},\ \bibinfo {pages} {042103} (\bibinfo {year}
  {2012})}\BibitemShut {NoStop}%
\bibitem [{\citenamefont {Watrous}(2002)}]{watrous2002quantum}%
  \BibitemOpen
  \bibfield  {author} {\bibinfo {author} {\bibfnamefont {John}\ \bibnamefont
  {Watrous}},\ }\bibfield  {title} {\enquote {\bibinfo {title} {Quantum
  statistical zero-knowledge},}\ }\href@noop {} {\bibfield  {journal} {\bibinfo
   {journal} {arXiv:quant-ph/0202111}\ } (\bibinfo {year} {2002})}\BibitemShut
  {NoStop}%
\bibitem [{\citenamefont {Weyl}(1912)}]{weyl1912asymptotische}%
  \BibitemOpen
  \bibfield  {author} {\bibinfo {author} {\bibfnamefont {H.}~\bibnamefont
  {Weyl}},\ }\bibfield  {title} {\enquote {\bibinfo {title} {Das asymptotische
  verteilungsgesetz der eigenwerte linearer partieller differentialgleichungen
  (mit einer anwendung auf die theorie der hohlraumstrahlung)},}\ }\href@noop
  {} {\bibfield  {journal} {\bibinfo  {journal} {Mathematische Annalen}\
  }\textbf {\bibinfo {volume} {71}},\ \bibinfo {pages} {441--479} (\bibinfo
  {year} {1912})}\BibitemShut {NoStop}%
\bibitem [{Note1()}]{Note1}%
  \BibitemOpen
  \bibinfo {note} {For example, the following bound would be natural
  conjecture: $D(\rho ,\sigma )^2 \leqslant D_{\protect \text {HS}}(\rho
  ,\sigma )/ {\protect \rm Tr}(\rho ^2)$. However, we numerically found this
  bound to be violated.}\BibitemShut {Stop}%
\bibitem [{\citenamefont {{Li}}\ and\ \citenamefont
  {{Haldane}}(2008)}]{li2008}%
  \BibitemOpen
  \bibfield  {author} {\bibinfo {author} {\bibfnamefont {H.}~\bibnamefont
  {{Li}}}\ and\ \bibinfo {author} {\bibfnamefont {F.~D.~M.}\ \bibnamefont
  {{Haldane}}},\ }\bibfield  {title} {\enquote {\bibinfo {title} {{Entanglement
  Spectrum as a Generalization of Entanglement Entropy: Identification of
  Topological Order in Non-Abelian Fractional Quantum Hall Effect States}},}\
  }\href {\doibase 10.1103/PhysRevLett.101.010504} {\bibfield  {journal}
  {\bibinfo  {journal} {Phys. Rev. Lett.}\ }\textbf {\bibinfo {volume} {101}},\
  \bibinfo {eid} {010504} (\bibinfo {year} {2008})}\BibitemShut {NoStop}%
\bibitem [{\citenamefont {Bishop}(2006)}]{bishop2006pattern}%
  \BibitemOpen
  \bibfield  {author} {\bibinfo {author} {\bibfnamefont {C.~M.}\ \bibnamefont
  {Bishop}},\ }\href@noop {} {\emph {\bibinfo {title} {Pattern recognition and
  machine learning}}}\ (\bibinfo  {publisher} {Springer},\ \bibinfo {year}
  {2006})\BibitemShut {NoStop}%
\bibitem [{\citenamefont {Murphy}(2012)}]{murphy2012machine}%
  \BibitemOpen
  \bibfield  {author} {\bibinfo {author} {\bibfnamefont {K.~P.}\ \bibnamefont
  {Murphy}},\ }\href@noop {} {\emph {\bibinfo {title} {Machine learning: a
  probabilistic perspective}}}\ (\bibinfo  {publisher} {MIT press},\ \bibinfo
  {year} {2012})\BibitemShut {NoStop}%
\end{thebibliography}

%

\end{document}